\documentclass[12pt,reqno]{amsart}
\allowdisplaybreaks[4]

\usepackage{graphicx}
\usepackage{amssymb}
\usepackage{color}
\usepackage{epstopdf}
\usepackage{amsthm,amsmath,amssymb}
\usepackage{enumitem}
\usepackage{cases}
\usepackage{mathrsfs}
\usepackage{cite}
\usepackage{subeqnarray}
\usepackage{cases}

\newtheorem{proposition}{Proposition}[section]

\newtheorem{thm}{Theorem}[section]
\newtheorem{definition}{Definition}[section]
\newtheorem{example}{Example}

\begin{document}

\begin{center}
{\large \sc \bf The (3+1)-dimensional dispersionless integrable hierarchy and nonlinear Riemann-Hilbert  problem associated with the Doubrov-Ferapontov modified heavenly equation  }

\vskip 20pt

{Ge Yi, Bowen Sun, Kelei Tian* and Ying Xu \\

\it
 Hefei University of Technology, Hefei, Anhui 230601, China}

\bigskip

\bigskip
$^*$ Corresponding author:  {\tt kltian@ustc.edu.cn, kltian@hfut.edu.cn}
\bigskip

\bigskip

{\today}

\end{center}

\bigskip
\bigskip
\textbf{Abstract:} According to the classification of integrable complex Monge–Ampère equations by Doubrov and Ferapontov,  the modified heavenly equation is a typical (3+1)-dimensional dispersionless and canonical integrable equation.
In this paper  we use the eigenfunctions of the Doubrov-Ferapontov modified heavenly equation to obtain a related hierarchy. 
Next we construct the Lax-Sato equations with Hamiltonian vector fields and Zakharov-Shabat type equations which are equivalent to the hierarchy.
The nonlinear Riemann-Hilbert  problem is also applied to study the solution of Doubrov-Ferapontov modified heavenly equation. 
\bigskip

\bigskip
\bigskip

\section{\sc \bf Introduction}
The study of integrable systems has a rich historical background. Integrable systems represent a significant research area within mathematical physics, intimately connected with both algebra and geometry. These include equations like  KdV (Korteweg-de Vries) equation,  KP (Kadomtsev-Petviashvili) equation,  Sine-Gordon equation and others.
The aforementioned integrable systems are commonly known as classical integrable systems. Another category of integrable systems, referred to as dispersionless integrable systems, encompasses equations such as the  dKP (dispersionless Kadomtsev-Petviashvili) equation, dDS (dispersionless Davey-Stewartson)  equation, the Pleba\'{n}ski second heavenly equation, Dunajski equation and others.  Research on these dispersionless integrable systems has captured the interest of many mathematicians and physicists.
According to the theory of Zakharov and Shabat, dispersionless integrable systems can be expressed through the commutation of Lax pairs of vector fields, enabling the construction of integrable systems in arbitrary dimensions, which distinguishes them from classical integrable systems\cite{Ref1979}.

The integrability of dispersionless integrable systems manifests through an infinite number of symmetries, which holds significant meaning throughout the entirety of the research, as defined by the hierarchy.
Takasaki,  who constructs the hyper-K$\ddot{a}$hler hierarchy, introduces a series of hidden independent variables which play the role of the time variables.  Apart from that, he also finds that the original field equations for 
 hyper-K$\ddot{a}$hler metric can be extended  to the hierarchy\cite{Ref1989,Ref1990}.
Then, many (2+1)-dimensional dispersionless integrable hierarchies have been studied, such as celebrated dKP (dispersionless Kadomtsev-Petviashvili) hierarchy \cite{Ref1995}, 
the Manakov-Santini hierarchy \cite{Ref2007}
and  dDS (dispersionless  Davey-Stewartson) hierarchy \cite{Ref2020,Ref2022}.
Above hierarchies are derived from (2+1)-dimensional integrable systems.
The (3+1)-dimensional dispersionless integrable systems can also construct the similar hierarchy.
Bogdanov and Konopelchenko develop the Pleba\'{n}ski second heavenly  hierarchy \cite{Ref2006}.
They consider the formal Laurent series 
$$
\begin{aligned}
&S^{1}=\sum_{n=0}^{\infty}t_{n}^{1}z^{n}+\sum_{n=1}^{\infty}S^{1}_{n}({\bf t}^{\bf 1},{\bf t}^{\bf2})z^{-n},
\\&S^{2}=\sum_{n=0}^{\infty}t_{n}^{2}z^{n}+\sum_{n=1}^{\infty}S^{2}_{n}({\bf t}^{\bf 1},{\bf t}^{\bf 2})z^{-n}.
\end{aligned}
$$
These formal series are eigenfunctions of the Lax pair of the Pleba\'{n}ski second heavenly equation.
They give the exterior differential forms of definition of the Pleba\'{n}ski second heavenly  hierarchy as 
$$(dS^{1}\wedge dS^{2})_{-}=0$$
and they also get the Lax-Sato form of the Pleba\'{n}ski second heavenly  hierarchy, which is defined as
$$
\begin{aligned}
&\partial_{n}^{1}\mathcal{S} =-\{(z^{n}S^{2})_{+},\mathcal{S}\},
\\&\partial_{n}^{2}\mathcal{S}=\{(z^{n}S^{1})_{+},\mathcal{S}\},
\\&\{S^{1},S^{2}\}=1,
\end{aligned}
$$
where $\{A,B\}=A_{x}B_{y}-A_{y}B_{x}$, ${\bf t}^{\bf 1}=(t_{0}^{1},t_{1}^{1},\dots,t_{n}^{1},\dots),  {\bf t}^{\bf 2}=(t_{0}^{2},t_{1}^{2},\dots,t_{n}^{2},\dots)$. Apart from that, they  introduce the projectors $\left(\sum\limits_{-\infty}^{\infty}u_{n}z^{n}\right)_{+}=\sum\limits_{n=0}^{\infty}u_{n}z^{n}$, $\left(\sum\limits_{-\infty}^{\infty}u_{n}z^{n}\right)_{-}=\sum\limits_{-\infty}^{n=-1}u_{n}z^{n}$.
These two definitions are equivalent and the Lax-Sato equations have Hamiltonian vector fields structure.
Besides,   Bogdanov,  Dryuma and  Manakov construct the Dunajski hierarchy by using the same method\cite{Ref2007a}.

In this paper,  we focus on the Doubrov-Ferapontov modified heavenly equation 
\begin{equation}
u_{zt}+u_{zx}u_{xy}-u_{yz}u_{xx}=0,
\end{equation}
which is a typical (3+1)-dimensional dispersionless integrable equation.
The equation (1) arises from the Lax pair
\begin{equation}
  \begin{split}
  T&=\partial_{t}+(u_{xy}-\lambda)\partial_{x}-u_{xx}\partial_{y},
  \\Z&=\lambda\partial_{z}-u_{yz}\partial_{x}+u_{zx}\partial_{y}.
  \end{split}
  \end{equation}
Doubrov and Ferapontov classify the Monge–Ampère type equations. Under their classification, equation (1) is  canonical\cite{Ref2010,Ref2014}.
Sheftel and Yazıcı have shown that the equation (1) is a particular case of their asymmetric heavenly equation
$$u_{zx}u_{xy}-u_{zy}u_{xx}+au_{xt}+bu_{zt}+cu_{zz}=0$$ at $a=c=0, b=1$.  They also get the anti-self-dual vacuum metric governed by solutions
of equation (1) which has the signature $(+,+,-,-)$. Apart from that, they also give symmetries admitted by the cubic solution of equation (1)\cite{Ref2014a,Ref2011}.
  Marvan and Sergyeyev present a new approach to construction of recursion operators for equation (1), where the recursion operators can generate an infinite hierarchy of symmetries \cite{Ref2012}.

The hierarchy linked to the equation (1) can be established through the definition of exterior differentials, akin to the second heavenly equation. This process yields the Lax-Sato equations, providing an equivalent definition for the hierarchy.  Apart from that, the Lax pair can be expressed in the following Hamiltonian vector fields form as
\begin{equation}
  \begin{split}
  T&=\partial_{t}-\{-\lambda y +u_{x},\cdot \},
  \\Z&=\lambda\partial_{z}-\{u_{z},\cdot \}.
  \end{split}
  \end{equation}
Therefore, in this article, we will also discuss the Hamiltonian vector fields of the  hierarchy and the Zakharov-Shabat type equations.

In the field related to dispersionless integrable equations, there exists  IST (inverse scattering transform)  named Manakov-Santini method.
Due to the fact that space of the eigenfunctions of the vector fields is a ring,  the inverse problem is essentially nonlinear and this nonlinear problem  can be formulated as a nonlinear Riemann-Hilbert problem
on a suitable contour of the complex plane of the spectral parameter\cite{Ref2006a,Ref2006b}. 
The novel IST has been applied for solving the dispersionless integrable equations, such as dKP equation\cite{Ref2006b}, dDS system\cite{Ref2020}, second heavenly equation\cite{Ref2006a,Ref2009} and Dunajski equations\cite{Ref2015}. Linking dispersionless integrable equations to the nonlinear Riemann-Hilbert problem is the first step in applying the IST, which is also the focus of our study in this article.

In section 2, we  construct a related hierarchy of Doubrov-Ferapontov modified heavenly equation.
In section 3, we  construct the Hamiltonian vector fields associated with the hierarchy and obtain Zakharov-Shabat type equations which are equivalent to the hierarchy.
In section 4, we  construct the pre-reduced hierarchy.
In section 5, we construct the related nonlinear Riemann-Hilbert problem for Doubrov-Ferapontov modified heavenly equation  with reality constraint and heavenly constraint.

\bigskip

\section{\sc \bf  The hierarchy associated with the Doubrov-Ferapontov modified heavenly equation}

The Lax pairs of the dispersionless integrable system  consists of vector fields.
The eigenfunctions space forms a ring (not only the sum, but also the product of eigenfunctions is also the eigenfunction) and its product is finite. 
Therefore Lax pair (2) have three eigenfunctions.
However, 
$\tilde{\phi_{0}}=\lambda$ is the simplest eigenfunction, so we only need to consider other two eigenfunctions $\tilde{\phi_{1}}$ and $\tilde{\phi_{2}}$.
 These  eigenfunctions are defined as the following formal Laurent series
\begin{align*}
&\tilde{\phi_{1}} =-y+\sum_{k=1}^{\infty}\frac{g_{k}}{\lambda^{k}},\\
&\tilde{\phi_{2}} =x+t\lambda +\sum_{m=1}^{\infty}\frac{f_{m}}{\lambda^{m}},
\end{align*}
where   
$\overset{\rightharpoonup }{\tilde{\phi} }=
\begin{gathered}
\begin{pmatrix}\tilde{\phi_{1}}  \\ \tilde{\phi_{2}} 
\end{pmatrix}
\end{gathered}, f_{m}=f_{m}(x,y,z,t), g_{k}=g_{k}(x,y,z,t)$.

From $T(\overset{\rightharpoonup }{\tilde{\phi} })=\overset{\rightharpoonup }{0}$ and $Z(\overset{\rightharpoonup }{\tilde{\phi}})=\overset{\rightharpoonup }{0}$,   we  compare the coefficients in front of $\lambda$ and get the following recursive relations as
$$g_{1}=u_{x}, g_{2}=u_{t},\dots,g_{k}=\partial_{z}^{-1}(u_{yz}g_{k-1,x}-u_{zx}g_{k-1,y}),\cdots$$
$$f_{1}=u_{y},f_{2}=\partial_{x}^{-1}(-u_{xx}f_{1,y}+u_{xy}f_{1,x}+f_{1,t}),\dots,f_{m}=\partial^{-1}_{x}(f_{m-1,t}+u_{xy}f_{m-1,x}-u_{xx}f_{m-1,y}),\cdots$$
To construct the hierarchy, we generalize the preceding  eigenfunctions as
\begin{align*}
&\phi_{1}=-y+\sum_{j=2}^{\infty}z_{j}\lambda^{j-1}+\sum_{k=1}^{\infty}\frac{g_{k}}{\lambda^{k}},\\
&\phi_{2}=x+t\lambda +\sum_{n=2}^{\infty}t_{n}\lambda^{n}+\sum_{m=1}^{\infty}\frac{f_{m}}{\lambda^{m}},
\end{align*}
where $\textbf{t}=(t_{0},t_{1},...,t_{n},...),$  $\textbf{z}=(z_{0},z_{1},..., z_{j},...),$   
$\overset{\rightharpoonup }{\phi} =
\begin{gathered}
\begin{pmatrix}\phi_{1}  \\ \phi_{2} 
\end{pmatrix}
\end{gathered}$,  
$f_{m}=f_{m}(\textbf{t},\textbf{z}),  g_{k}=g_{k}(\textbf{t},\textbf{z})$. In particular, $t_{0}=x, t_{1}=t, z_{0}=y, z_{1}=z.$

\begin{definition}
The  Doubrov-Ferapontov modified heavenly  hierarchy is defined as
\begin{equation}
\left(\Phi_{1}\wedge \Phi_{2}\right)_{-}=0,
\end{equation}
where 
\begin{equation}
\Phi=\phi_{x}dx+\phi_{y}dy+\sum_{n=1}^{\infty}\phi_{t_{n}}dt_{n}+\lambda \sum_{j=1}^{\infty}\phi_{z_{j}}dz_{j}.
\end{equation}
\end{definition}
According to the identities (4) and (5), we only need to consider the coefficients of positive powers of $\lambda$. Therefore we get the following 2-form exterior differential as
\begin{align}
  \nonumber
&(\Phi_{1}\wedge \Phi_{2})_{+}
  \nonumber
\\&=\left(\frac{\partial(\phi_{1},\phi_{2})}{\partial(x,y)}\right)_{+}dx\wedge dy+\sum_{n=1}^{\infty}
  \left(\frac{\partial(\phi_{1},\phi_{2})}{\partial(t_{n},x)}\right)_{+}dt_{n}\wedge dx
+\sum_{n=1}^{\infty}\left(\frac{\partial(\phi_{1},\phi_{2})}{\partial(t_{n},y)}\right)_{+}dt_{n}\wedge dy
\nonumber
\\&+\sum_{j=1}^{\infty}\left(\frac{\lambda \partial(\phi_{1},\phi_{2})}{\partial(z_{j},x)}\right)_{+}dz_{j}\wedge dx
+\sum_{j=1}^{\infty}\left(\frac{\lambda \partial(\phi_{1},\phi_{2})}{\partial(z_{j},y)}\right)_{+}dz_{j}\wedge dy
\nonumber
\\&+\sum_{n=1}^{\infty}\sum_{n \textless s}\left(\frac{\partial(\phi_{1},\phi_{2})}{\partial(t_{n},t_{s})}\right)_{+}dt_{n}\wedge dt_{s}
\nonumber
+\sum_{n=1}^{\infty}\sum_{j=1}^{\infty}\left(\frac{\lambda \partial(\phi_{1},\phi_{2})}{\partial(t_{n},z_{j})}\right)_{+}dt_{n}\wedge dz_{j}
\\&+\sum_{j=1}^{\infty}\sum_{j\textless i}\left(\frac{\lambda^{2} \partial(\phi_{1},\phi_{2})}{\partial(z_{j},z_{i})}\right)_{+}dz_{j}\wedge dz_{i}.
\end{align}
Then we introduce the matrix
\begin{align*}
J=
\begin{gathered}
\begin{pmatrix}
\phi_{1,x}& \phi_{2,x} \\ \phi_{1,y}&\phi_{2,y}
\end{pmatrix}
\end{gathered}
\end{align*}
and we find that $detJ$ is equivalent to the coefficient of $dx\wedge dy$.
Thus we can calculate that
\begin{align*}
\nonumber
detJ&=
\begin{vmatrix}
\phi_{1,x}&\phi_{2,x}\\
\phi_{1,y}&\phi_{2,y}
\end{vmatrix}_{+}
\nonumber
\\&=\begin{vmatrix}
\sum\limits_{k=1}^{\infty}g_{k,x}\lambda^{-k}&1+\sum\limits_{m=1}^{\infty}f_{m,x}\lambda^{-m}\\
-1+\sum\limits_{k=1}^{\infty}g_{k,x}\lambda^{-k}&\sum\limits_{m=1}^{\infty}f_{m,y}\lambda^{-m}
\end{vmatrix}_{+}
\nonumber
\\&=1
\\&=\begin{vmatrix}
\phi_{1,x}&\phi_{1,y}\\
\phi_{2,x}&\phi_{2,y}
\end{vmatrix}_{+}
=\left(\frac{\partial(\phi_{1},\phi_{2})}{\partial(x,y)}\right)_{+}
\end{align*}
and  also get the adjoint matrix as
\begin{align*}
J^{*}=
\begin{gathered}
\begin{pmatrix}
\phi_{2,y}& -\phi_{2,x} \\ -\phi_{1,y}&\phi_{1,x}
\end{pmatrix}
\end{gathered}.
\end{align*}
Next we will consider the coefficients in front of $dt_{n}\wedge dx,dt_{n}\wedge dy,dz_{j}\wedge dx$ and $dz_{j}\wedge dy$.
We will construct a series of operators using these coefficients, which satisfy commutation conditions, and form a hierarchy.
In first situation, 

\begin{align*}
\begin{gathered}
    \begin{pmatrix}
    \phi_{1,t_{n}} ,& \phi_{2,t_{n}}
    \end{pmatrix}
\end{gathered}J^{*}&=
\begin{gathered}
\begin{pmatrix}
 \phi_{1,t_{n}} ,&  \phi_{2,t_{n}}
\end{pmatrix}
\end{gathered}
\begin{gathered}
  \begin{pmatrix}
  \phi_{2,y}& -\phi_{2,x} \\ -\phi_{1,y}&\phi_{1,x}
  \end{pmatrix}
\end{gathered}
=
\begin{gathered}
\begin{pmatrix}
X_{n},&Y_{n}
\end{pmatrix}
\end{gathered}
\end{align*}
so 
 \begin{align}
 \begin{gathered}
  \begin{pmatrix}
   \phi_{1,t_{n}} ,&  \phi_{2,t_{n}}
  \end{pmatrix}
\end{gathered}&=
\begin{gathered}
  \begin{pmatrix}
  X_{n},&Y_{n}
  \end{pmatrix}
  \end{gathered}
  \begin{gathered}
    \begin{pmatrix}
    \phi_{1,x}& \phi_{2,x} \\ \phi_{1,y}&\phi_{2,y}
    \end{pmatrix}
  \end{gathered}
=\begin{gathered}
  \begin{pmatrix}
  X_{n},&Y_{n}
  \end{pmatrix}
  \end{gathered}J.
\end{align}
Therefore, according to the identity (7), we can get
\begin{equation}
\frac{\partial \overset{\rightharpoonup }{\phi}}{\partial{t_{n}}}= (X_{n}\partial_{x}+Y_{n}\partial_{y})( \overset{\rightharpoonup }{\phi}).  
\end{equation}
Due to identities (4) and (5),  the negative power of $\lambda$ vanishes, so
\begin{align*}
&X_{n}=
\begin{vmatrix}\phi_{1,t_{n}}& \phi_{2,t_{n}} \\ \phi_{1,y}&\phi_{2,y}
\end{vmatrix}
=
\begin{vmatrix}\phi_{1,t_{n}}& \phi_{2,t_{n}} \\ \phi_{1,y}&\phi_{2,y}
\end{vmatrix}_{+}
=\left(\frac{\partial(\phi_{1},\phi_{2})}{\partial(t_{n},y)}\right)_{+},
\\&Y_{n}=\begin{vmatrix}\phi_{2,t_{n}}& \phi_{1,t_{n}} \\ \phi_{2,x}&\phi_{1,x}\end{vmatrix}=\begin{vmatrix}\phi_{2,t_{n}}& \phi_{1,t_{n}} \\ \phi_{2,x}&\phi_{1,x}\end{vmatrix}_{+}
=\left(-\frac{\partial(\phi_{1},\phi_{2})}{\partial(t_{n},x)}\right)_{+}.
\end{align*}
Therefore
\begin{equation*}
\frac{\partial  \overset{\rightharpoonup }{\phi}}{\partial{t_{n}}}=
\left(
\begin{gathered}
\begin{vmatrix}\phi_{1,t_{n}}& \phi_{2,t_{n}} \\ \phi_{1,y}&\phi_{2,y}
\end{vmatrix}
\end{gathered}_{+}\partial_{x}+
\begin{gathered}
\begin{vmatrix}\phi_{2,t_{n}}& \phi_{1,t_{n}} \\ \phi_{2,x}&\phi_{1,x}
\end{vmatrix}
\end{gathered}_{+}\partial_{y}\right)( \overset{\rightharpoonup }{\phi}).
\end{equation*}
Similarly,  another vector field can be defined. At first,
\begin{align*}
  \begin{gathered}
      \begin{pmatrix}
      \lambda  \phi_{1,z_{j}} ,& \lambda  \phi_{2,z_{j}}
      \end{pmatrix}
  \end{gathered}J^{*}=
  \begin{gathered}
    \begin{pmatrix}
    \lambda  \phi_{1,z_{j}} ,& \lambda  \phi_{2,z_{j}}
    \end{pmatrix}
  \end{gathered}
    \begin{gathered}
      \begin{pmatrix}
      \phi_{2,y}& -\phi_{2,x} \\ -\phi_{1,y}&\phi_{1,x}
      \end{pmatrix}
      \end{gathered}
    =(\tilde{X_{j}},\tilde{Y_{j}})
  \end{align*}
  and
  \begin{equation}
\frac{\lambda \partial \overset{\rightharpoonup }{\phi}}{\partial{z_{j}}}= (\tilde{X_{j}}\partial_{x}+\tilde{Y_{j}}\partial_{y})( \overset{\rightharpoonup }{\phi}), 
\end{equation}
  so
  \begin{align}
    &\tilde{X_{j}}=\begin{vmatrix}\lambda\phi_{1,z_{j}}& \lambda\phi_{2,z_{j}} \\ \phi_{1,y}&\phi_{2,y}\end{vmatrix}=\begin{vmatrix}\lambda\phi_{1,z_{j}}& \lambda\phi_{2,z_{j}} \\ \phi_{1,y}&\phi_{2,y}\end{vmatrix}_{+}
    =\left(\frac{\lambda \partial(\phi_{1},\phi_{2})}{\partial(z_{j},y)}\right)_{+},
  \\&\tilde{Y_{j}}=\begin{vmatrix}\lambda\phi_{2,z_{j}}& \lambda\phi_{1,z_{j}} \\ \phi_{2,x}&\phi_{1,x}\end{vmatrix}=\begin{vmatrix}\lambda\phi_{2,z_{j}}& \lambda\phi_{1,z_{j}} \\ \phi_{2,x}&\phi_{1,x}\end{vmatrix}_{+}
  =\left(-\frac{\lambda \partial(\phi_{1},\phi_{2})}{\partial(z_{j},x)}\right)_{+}.
  \end{align}
  Finally, because of the identities (9), (10) and (11), the  equation is given as
  \begin{equation*}
\frac{\lambda\partial  \overset{\rightharpoonup }{\phi}}{\partial z_{j}}=\left(
  \begin{vmatrix}\lambda\phi_{1,z_{j}}& \lambda\phi_{2,z_{j}} \\ \phi_{1,y}&\phi_{2,y}\end{vmatrix}_{+}\partial_{x}
  +\begin{vmatrix}\lambda\phi_{2,z_{j}}& \lambda\phi_{1,z_{j}} \\ \phi_{2,x}&\phi_{1,x}\end{vmatrix}_{+}\partial_{y}\right)( \vec{\phi}).
\end{equation*}

\begin{thm}
All coefficient functions $g_{k}$ and $f_{m}$ can solve the hierarchy $(4)$ if and only if  they satisfy following Lax-Sato form equations as
\begin{align}
&\frac{\partial  \overset{\rightharpoonup }{\phi}}{\partial{t_{n}}}=
\left(
\begin{gathered}
\begin{vmatrix}\phi_{1,t_{n}}& \phi_{2,t_{n}} \\ \phi_{1,y}&\phi_{2,y}
\end{vmatrix}
\end{gathered}_{+}\partial_{x}+
\begin{gathered}
\begin{vmatrix}\phi_{2,t_{n}}& \phi_{1,t_{n}} \\ \phi_{2,x}&\phi_{1,x}
\end{vmatrix}
\end{gathered}_{+}\partial_{y}\right)( \overset{\rightharpoonup }{\phi}),
\\&\frac{\lambda\partial  \vec{\phi}}{\partial z_{j}}=\left(
  \begin{vmatrix}\lambda\phi_{1,z_{j}}& \lambda\phi_{2,z_{j}} \\ \phi_{1,y}&\phi_{2,y}\end{vmatrix}_{+}\partial_{x}
  +\begin{vmatrix}\lambda\phi_{2,z_{j}}& \lambda\phi_{1,z_{j}} \\ \phi_{2,x}&\phi_{1,x}\end{vmatrix}_{+}\partial_{y}\right)( \overset{\rightharpoonup }{\phi}).
\end{align}
\end{thm}
\begin{proof}
On the one hand,
according to the preceding calculation, we can get  Lax-Sato equations (12) and (13).

On the other hand, 
for each coefficients of 2-form of $\Phi_{1}\wedge \Phi_{2}$, they have the same validation method. So 
we only display the coefficients of $dt_{n} \wedge dz_{j}$ as
$$
\begin{aligned}
\frac{\lambda \partial(\phi_{1},\phi_{2})}{\partial(t_{n},z_{j})}
&=\begin{vmatrix}\phi_{1,t_{n}}& \phi_{2,t_{n}} \\ \lambda\phi_{1,z_{j}}&\lambda\phi_{2,z_{j}}\end{vmatrix}
=\begin{vmatrix}[X_{n}]_{+}\phi_{1,y}+[Y_{n}]_{+}\phi_{1,x}& [X_{n}]_{+}\phi_{2,x}+[Y_{n}]_{+}\phi_{2,y} 
  \\ [\tilde{X_{j}}]_{+}\phi_{1,x}+[\tilde{Y_{j}}]_{+}\phi_{1,y}&[\tilde{X_{j}}]_{+}\phi_{2,x}+[\tilde{Y_{j}}]_{+}\phi_{2,y}\end{vmatrix}
\\&=\begin{vmatrix}[X_{n}]_{+}& [Y_{n}]_{+} \\ [\tilde{X_{n}}]_{+}&[\tilde{Y_{n}}]_{+}\end{vmatrix}.
\begin{vmatrix}
  \phi_{1,x}&\phi_{2,x}\\
  \phi_{1,y}&\phi_{2,y}
  \end{vmatrix}
=\begin{vmatrix}[X_{n}]_{+}& [Y_{n}]_{+} \\ [\tilde{X_{n}}]_{+}&[\tilde{Y_{n}}]_{+}\end{vmatrix}.
\end{aligned}
$$
So 
$$\left(\frac{\lambda \partial(\phi_{1},\phi_{2})}{\partial(t_{n},z_{j})}\right)_{-}=0.$$
When we apply this method to other coefficients, we can get the identities (4).
\end{proof}
Then we let 
\begin{equation}
T_{n}=\partial_{t_{n}}-
\begin{gathered}
\begin{vmatrix}\phi_{1,t_{n}}& \phi_{2,t_{n}} \\ \phi_{1,y}&\phi_{2,y}
\end{vmatrix}
\end{gathered}_{+}\partial_{x}
-
\begin{gathered}
\begin{vmatrix}\phi_{2,t_{n}}& \phi_{1,t_{n}} \\ \phi_{2,x}&\phi_{1,x}
\end{vmatrix}
\end{gathered}_{+}\partial_{y}
\end{equation}
and 
\begin{equation}
Z_{j}=\lambda\partial_{z_{j}}-
\begin{gathered}
\begin{vmatrix}\lambda \phi_{1,z_{j}}& \lambda \phi_{2,z_{j}} \\ \phi_{1,y}&\phi_{2,y}
\end{vmatrix}
\end{gathered}_{+}\partial_{x}
-
\begin{gathered}
\begin{vmatrix}\lambda \phi_{2,z_{j}}&\lambda \phi_{1,z_{j}} \\ \phi_{2,x}&\phi_{1,x}
\end{vmatrix}
\end{gathered}_{+}\partial_{y},
\end{equation}
we can get
\begin{equation*}
T_{n}( \overset{\rightharpoonup }{\phi})=\overset{\rightharpoonup }{0}
\end{equation*}
and
\begin{equation*}
  Z_{j}( \overset{\rightharpoonup }{\phi})=\overset{\rightharpoonup }{0}.
\end{equation*}
According to above identities,  $T_{n}$ and $Z_{j}$ share the eigenfunction space. Thus  the commutation condition of $T_{n}$ and  $Z_{j}$ can be expressed as
\begin{equation*}
[T_{n},T_{s}]=0,
\end{equation*}
\begin{equation*}
[T_{n},Z_{j}]=0
\end{equation*}
and
\begin{equation*}
[Z_{i},Z_{j}]=0.
\end{equation*}
In the last part of this section, according to previous identities, we can derive some structurally interesting dispersionless integrable equations. These equations are novel within the scope of our current knowledge. At first, we provide some specific examples of (14) and (15) as
\begin{align*}
   &T_{1}= \partial_{t_{1}}-(\lambda - u_{xy})\partial_{x}-u_{xx}\partial_{y},
   \\&T_{2}=\partial_{t_{2}}-(\lambda^{2}-\lambda u_{xy}-g_{2,y})\partial_{x}-(\lambda u_{xx}+g_{2,x})\partial_{y},
\\&T_{3}=\partial_{t_{3}}-(\lambda^{3}-\lambda^{2}u_{xy}-\lambda u_{yt}-g_{3,y})\partial_{x}-(\lambda^{2} u_{xx}+\lambda u_{xt}+g_{3,x})\partial_{y},
 \\&Z_{1}=\lambda \partial_{z_{1}}-u_{yz}\partial_{x}+u_{zx}\partial_{y},
\\&Z_{2}=\lambda \partial_{z_{2}}-(\lambda u_{yy}+f_{2,y}+u_{yz_{2}})\partial_{x}+(\lambda^{2}+\lambda u_{xy}+f_{2,x}+u_{xz_{2}})\partial_{y}.
  \end{align*}
Then we give some examples of the commutation condition of these operators.
\begin{example}
\rm
 When $n=1$, $j=1$, 
  let $t_{1}=t,z_{1}=z$, we can find that $T_{1}$ and $Z_{1}$ are the original Lax pairs of the Doubrov-Ferapontov modified heavenly equation
 and the commutation condition of $T_{1}, Z_{1}$ writes as equation (1).
  \end{example}
\begin{example}
\rm
  When $n=1$, $s=2$,
The commutation condition of $T_{1}, T_{2}$ writes as
  \begin{align*}
  u_{xt_{2}}+u_{xx}u_{yt}-u_{xt}u_{xy}-u_{tt}=0.
  \end{align*}
\end{example}
\begin{example}
\rm
  When $n=2$, $j=1$, The commutation condition of $T_{2}, Z_{1}$ writes as
  \begin{align*}
    &u_{zt_{2}}-u_{yz}u_{xt}+u_{xz}u_{yt}=0.
  \end{align*} 
\end{example}

\begin{example}
\rm
When $j=2$, we can get 
The commutation condition of $Z_{1}, Z_{2}$ writes as
\begin{align*}
  &f_{2,y}u_{xz}+u_{yz_{2}}u_{xz}-f_{2,x}u_{yz}-u_{yz}u_{xz_{2}}=0,
  \\&f_{2,z}+u_{xz}u_{yy}-u_{yz}u_{xz_{2}}=0.
\end{align*}
\end{example}

\begin{example}
\rm
   When $n=3$,
The commutation condition of $T_{1}, T_{3}$ writes as
\begin{align*}
  &g_{3,x}+u_{xx}u_{yt}-u_{xt}u_{xy}-u_{tt}=0,
  \\&u_{xt_{3}}+u_{xx}u_{yt_{2}}-u_{xt_{2}}u_{xy}-g_{3,t}=0.
\end{align*}
\end{example}

\begin{example}
\rm
The commutation condition of $T_{2}, Z_{2}$ writes as
\begin{align*}
  &g_{2,z_{2}}+u_{yt_{2}}+u_{xy}u_{xz_{2}}+u_{xy}f_{2,x}+u_{xy}g_{2,y}-f_{2,y}u_{xx}-u_{yy}g_{2,x}-u_{yz_{2}}u_{xx}=0,
  \\&f_{2,t_{2}}+u_{z_{2}t_{2}}+f_{2,x}g_{2,y}+u_{xz_{2}}g_{2,y}-f_{2,y}g_{2,x}-u_{yz_{2}}g_{2,x}=0,
\\&f_{2,x}-g_{2,y}+u_{xx}u_{yy}-u_{xy}u_{xy}=0.
\end{align*}
\end{example}
\section{\sc \bf Hamiltonian vector fields and Zakharov-Shabat type equations}
According to the theory of dispersionless integrable hierarchy, when the Lax pairs are Hamiltonian vector fields, the hierarchy can be expressed not only as the Lax-Sato equations but also as the Zakharov-Shabat equations. For example, the dKP hierarchy and dDS hierarchy can be written in terms of the Zakharov-Shabat equations. Although the Lax pairs of the  Doubrov-Ferapontov modified heavenly  hierarchy are Hamiltonian vector fields, unlike previous Zakharov-Shabat equations, the Zakharov-Shabat equations we construct in this paper are not equal to zero but instead involve powers of $\lambda$. 
Besides, the space of eigenfunctions is also a Lie algebra, whose commutator is given by the Poisson bracket.
Therefore, we refer to the constructed equations as Zakharov-Shabat type equations.

At first,
\begin{align*}
\nonumber
X_{n}
&=
\begin{vmatrix}\phi_{1,t_{n}}& \phi_{2,t_{n}} \\ \phi_{1,y}&\phi_{2,y}
\end{vmatrix}_{+}
=\begin{vmatrix}\sum\limits_{k=1}^{\infty}\frac{g_{k,t_{n}}}{\lambda^{k}}&\lambda^{n}+ \sum\limits_{m=1}^{\infty}\frac{f_{m,t_{n}}}{\lambda^{m}} \\ -1+\sum\limits_{k=1}^{\infty}\frac{g_{k,y}}{\lambda^{k}}&\sum\limits_{m=1}^{\infty}\frac{f_{m,y}}{\lambda^{m}}
\end{vmatrix}_{+}
\\&=\lambda^{n}-\lambda^{n}\sum\limits_{k=1}^{n}\frac{g_{k,y}}{\lambda^{k}},
\end{align*}
\begin{align*}
\nonumber
Y_{n}&=-\begin{vmatrix}\phi_{1,t_{n}}& \phi_{2,t_{n}} \\ \phi_{1,x}&\phi_{2,x}\end{vmatrix}_{+}
=-\begin{vmatrix}\sum\limits_{k=1}^{\infty}\frac{g_{k,t_{n}}}{\lambda^{k}}&\lambda^{n}+ \sum\limits_{m=1}^{\infty}\frac{f_{m,t_{n}}}{\lambda^{m}} \\ \sum\limits_{k=1}^{\infty}\frac{g_{k,x}}{\lambda^{k}}&1+\sum\limits_{m=1}^{\infty}\frac{f_{m,x}}{\lambda^{m}}
\end{vmatrix}_{+}
\\&=\lambda^{n}\sum\limits_{k=1}^{n}\frac{g_{k,x}}{\lambda^{k}}
\end{align*}
Then, we can express the Lax-Sato equation (12) as 
\begin{equation}
\frac{\partial  \overset{\rightharpoonup }{\phi}}{\partial{t_{n}}}=\{-\lambda^{n}y+A_{n},\overset{\rightharpoonup }{\phi}\},
\end{equation}
where $$A_{n}=\lambda^{n}\sum\limits_{k=1}^{n}\frac{g_{k}}{\lambda^{k}}.$$
Similarly, the Lax-Sato equation (13) can be expressed as
\begin{equation}
\frac{\lambda\partial  \overset{\rightharpoonup }{\phi}}{\partial{z_{j}}}=\{-\lambda^{j}x+B_{j},\overset{\rightharpoonup }{\phi}\}, (j\geq 2)
\end{equation}
where $$B_{j}=-\lambda^{j}\sum\limits_{m=1}^{j}\frac{f_{m}}{\lambda^{m}}-u_{z_{j}},$$
and
\begin{equation}
\frac{\lambda\partial  \overset{\rightharpoonup }{\phi}}{\partial{z_{j}}}=\{B_{j},\overset{\rightharpoonup }{\phi}\}, (j=1),
\end{equation}
where
$$B_{1}=-u_{z}.$$
Therefore, according to the identities (16), (17) and (18), we can get Lax-Sato equations with Hamiltonian vector fields and construct the Zakharov-Shabat type equations.
To find Zakharov-Shabat type equations, we need to specify the relationship between the coefficient functions $g_{n+1}$, $f_{j}$ and the unknown function u.
\begin{proposition}
The relationship between the coefficient functions $g_{n+1}$, $f_{j}$ and the unknown function u are
\begin{equation}
g_{n+1}=u_{t_{n}}
\end{equation}
and
\begin{equation}
f_{j,t}=u_{tz_{j}}+f_{j,y}u_{xx}+u_{yz_{j}}u_{xx}-f_{j,y}u_{xy}-u_{xz_{j}}u_{xy}.
\end{equation}
\end{proposition}
In fact, the  general case of $T_{n}$ is given as
$$T_{n}=\partial_{t_{n}}-(\lambda^{n}-\lambda^{n-1}g_{1,y}-\cdots-g_{n,y})\partial_{x}-(\lambda^{n-1}g_{1,x}+\cdots+g_{n,y})\partial_{y}.$$
The commutation condition of $T_{n}$ and $Z_{1}$ writes as
$$u_{zt_{n}}+u_{xz}g_{n,y}-u_{yz}g_{n,x}=0.$$
According to the identity
$$g_{k}=\partial_{z}^{-1}(u_{yz}g_{k-1,x}-u_{zx}g_{k-1,y}),$$
we can find $$g_{n+1}=u_{t_{n}}.$$
Similarly, the general case of $Z_{j}$ is calculated as
$$Z_{j}=\lambda\partial_{z_{j}}-(\lambda^{j-1}f_{1,y}+\lambda^{j-2}f_{2,y}+\cdots+f_{j,y}+u_{yz_{j}})\partial_{x}+(\lambda^{j}+\lambda^{j-1}f_{1,x}+\lambda^{j-2}f_{2,x}+\cdots+f_{j,x}+u_{xz_{j}})\partial_{y}.$$
The commutation condition of $Z_{j}$ and $T_{1}$ writes as
$$f_{j,t}=-u_{tz_{j}}+f_{j,y}u_{xx}+u_{yz_{j}}u_{xx}-f_{j,x}u_{xy}-u_{xz_{j}}u_{xy}.$$
\begin{thm}
The hierarchy  is  equivalent to the following Zakharov-Shabat type equations as
\begin{subequations}
\begin{align}
&\frac{\partial A_{n}}{\partial t_{s}}-\frac{\partial A_{s}}{\partial t_{n}}+\{A_{n},A_{s}\}=-\lambda^{n}A_{s,x}+\lambda^{s}A_{n,x},
\\&\frac{\lambda\partial A_{n}}{\partial z}-\frac{\partial B_{1}}{\partial t_{n}}+\{A_{n},B_{1}\}=-\lambda^{n}B_{1,x}, 
\\&\frac{\lambda\partial A_{n}}{\partial z_{j}}-\frac{\partial B_{j}}{\partial t_{n}}+\{A_{n},B_{j}\}=-\lambda^{j}A_{n,y}-\lambda^{n}B_{j,x},j\geq2,
\\&\frac{\lambda\partial B_{j}}{\partial z}-\frac{\lambda\partial B_{1}}{\partial z_{j}}+\{B_{j},B_{1}\}=\lambda^{j}B_{1,y}, j\geq2,
\\&\frac{\lambda\partial B_{i}}{\partial z_{j}}-\frac{\lambda\partial B_{j}}{\partial z_{i}}+\{B_{i},B_{j}\}=-\lambda^{j}B_{i,y}+\lambda^{i}B_{j,y},i\geq2, j\geq2,
\end{align}
\end{subequations}
\end{thm}
\begin{proof}
In first situation, equation (21a) will be proved.
According to the identities (6), we consider the coefficients of $dt_{n}\wedge dt_{s}.$
\begin{align*}
\begin{vmatrix} \phi_{1,t_{n}}&\phi_{2,t_{n}} \\ \phi_{1,t_{s}}&\phi_{2,t_{s}}
\end{vmatrix}_{+}&=
\begin{vmatrix}\sum\limits_{k=1}^{\infty}g_{k,t_{n}}\lambda^{-k}&\lambda^{n}+\sum\limits_{m=1}^{\infty}f_{m,t_{n}}\lambda^{-m} \\ \sum\limits_{k=1}^{\infty}g_{k,t_{s}}\lambda^{-k}&\lambda^{s}+\sum\limits_{m=1}^{\infty}f_{m,t_{s}}\lambda^{-m}
\end{vmatrix}_{+}
\\&=-\lambda^{n}\sum\limits_{k=1}^{n}\frac{g_{k,t_{s}}}{\lambda^{k}}+\lambda^{s}\sum\limits_{k=1}^{s}\frac{g_{k,t_{n}}}{\lambda^{k}}
\\&=\begin{vmatrix}\{-\lambda^{n}y+A_{n},\phi_{1}\}&\{-\lambda^{n}y+A_{n},\phi_{2}\} \\ \{-\lambda^{s}y+A_{s},\phi_{1}\}&\{-\lambda^{s}y+A_{s},\phi_{2}\}
\end{vmatrix}_{+}
\\&=\lambda^{n}A_{s,x}-\lambda^{s}A_{n,x}+\{A_{n},A_{s}\}.
\end{align*}
Due to the identity$$\frac{\partial A_{s}}{\partial t_{n}}-\frac{\partial A_{n}}{\partial t_{s}}=-\lambda^{n}\sum\limits_{k=1}^{n}\frac{g_{k,t_{s}}}{\lambda^{k}}+\lambda^{s}\sum\limits_{k=1}^{s}\frac{g_{k,t_{n}}}{\lambda^{k}},$$ 
so
$$\frac{\partial A_{n}}{\partial t_{s}}-\frac{\partial A_{s}}{\partial t_{n}}+\lambda^{n}A_{s,x}-\lambda^{s}A_{n,x}+\{A_{n},A_{s}\}=0.$$
On the other hand, 
\begin{align*}
[T_{n},T_{s}]&=[\partial_{t_{n}}-\{-\lambda^{n}y+A_{n},\cdot\}, \partial_{t_{s}}-\{-\lambda^{s}y+A_{s},\cdot\}]
\\&=\left(\frac{\partial A_{n}}{\partial t_{s}}-\frac{\partial A_{s}}{\partial t_{n}}+\lambda^{n}A_{s,x}-\lambda^{s}A_{n,x}+\{A_{n},A_{s}\}\right)_{x}\partial_{y}
\\&-\left(\frac{\partial A_{n}}{\partial t_{s}}-\frac{\partial A_{s}}{\partial t_{n}}+\lambda^{n}A_{s,x}-\lambda^{s}A_{n,x}+\{A_{n},A_{s}\}\right)_{y}\partial_{x}
\\&=0.
\end{align*}
Therefore the commutation condition of $T_{n},T_{s}$ is proved.

In second situation, the coefficients of $dz_{j}\wedge dt_{n}$ and $dt_{n}\wedge dz_{1}$ should be studied.
We first consider the coefficients of $dz_{j}\wedge dt_{n}.$  On the one hand,
\begin{align*}
\begin{vmatrix}\lambda \phi_{1,z_{j}}&\lambda \phi_{2,z_{j}} \\ \phi_{1,t_{n}}&\phi_{2,t_{n}}
\end{vmatrix}_{+}&=
\begin{vmatrix}\lambda^{j}+\lambda\sum\limits_{k=1}^{\infty}g_{k,z_{j}}\lambda^{-k}& \lambda\sum\limits_{m=1}^{\infty}f_{m,z_{j}}\lambda^{-m} \\ \sum\limits_{k=1}^{\infty}g_{k,t_{n}}\lambda^{-k}&\lambda^{n}+\sum\limits_{m=1}^{\infty}f_{m,t_{n}}\lambda^{-m}
\end{vmatrix}_{+}
\\&=\lambda^{n+j}+\lambda^{j}\sum\limits_{m=1}^{j}\frac{f_{m,t_{n}}}{\lambda^{m}}+\lambda^{n+1}\sum\limits_{k=1}^{n+1}\frac{g_{k,z_{j}}}{\lambda^{k}}
\\&=\begin{vmatrix}\{-\lambda^{j}x+B_{j},\phi_{1}\}&\{-\lambda^{j}x+B_{j},\phi_{2}\} \\ \{-\lambda^{n}y+A_{n},\phi_{1}\}&\{-\lambda^{n}y+A_{n},\phi_{2}\}
\end{vmatrix}_{+}
\\&=\lambda^{n+j}-\lambda^{j}A_{n,y}-\lambda^{n}B_{j,x}-\{A_{n},B_{j}\}.
\end{align*}
Therefore, we can get the equation
\begin{equation*}
\lambda^{j}\sum\limits_{m=1}^{j}\frac{f_{m,t_{n}}}{\lambda^{m}}+\lambda^{n+1}\sum\limits_{k=1}^{n+1}\frac{g_{k,z_{j}}}{\lambda^{k}}+\lambda^{j}A_{n,y}+\lambda^{n}B_{j,x}+\{A_{n},B_{j}\}=0.
\end{equation*}
According to the identity (19), we can find
$$\lambda^{n+1}\sum\limits_{k=1}^{n+1}\frac{g_{k,z_{j}}}{\lambda^{k}}=\lambda^{n+1}\sum\limits_{k=1}^{n}\frac{g_{k,z_{j}}}{\lambda^{k}}+g_{n+1,z_{j}}=\lambda^{n+1}\sum\limits_{k=1}^{n}\frac{g_{k,z_{j}}}{\lambda^{k}}+u_{t_{n},z_{j}}=\frac{\lambda\partial{A_{n}}}{\partial z_{j}}+u_{t_{n},z_{j}}$$ and 
$$\frac{\lambda\partial A_{n}}{\partial z_{j}}-\frac{\partial B_{j}}{\partial t_{n}}+\lambda^{j}A_{n,y}+\lambda^{n}B_{j,x}+\{A_{n},B_{j}\}=0,$$
which is equation (21c). 
On the other hand, 
\begin{align*}
[T_{n},Z_{j}]&=[\partial_{t_{n}}-\{-\lambda^{n}y+A_{n},\cdot\},\lambda\partial_{z_{j}}-\{-\lambda^{j}x+B_{j},\cdot\}]
\\&=\left(\frac{\lambda\partial A_{n}}{\partial z_{j}}-\frac{\partial B_{j}}{\partial t_{n}}+\lambda^{j}A_{n,y}+\lambda^{n}B_{j,x}+\{A_{n},B_{j}\}\right)_{x}\partial_{y}
\\&-\left(\frac{\lambda\partial A_{n}}{\partial z_{j}}-\frac{\partial B_{j}}{\partial t_{n}}+\lambda^{j}A_{n,y}+\lambda^{n}B_{j,x}+\{A_{n},B_{j}\}\right)_{y}\partial_{x}
\\&=0.
\end{align*}
Similarly, we can get equation (21b) and its equivalent commutation condition of $T_{n},Z_{1}$.

In the final situation, the coefficients of $dz_{j}\wedge dz_{i}$ and $dz_{j}\wedge dz_{1}$ will be focused on. Firstly, we consider the  coefficients of $dz_{j}\wedge dz_{i}.$ On the one hand,
\begin{align*}
\begin{vmatrix}\lambda \phi_{1,z_{j}}&\lambda \phi_{2,z_{j}} \\ \lambda \phi_{1,z_{i}}&\lambda \phi_{2,z_{i}}
\end{vmatrix}_{+}&=
\begin{vmatrix}\lambda^{j}+\lambda\sum\limits_{k=1}^{\infty}g_{k,z_{j}}\lambda^{-k}& \lambda\sum\limits_{m=1}^{\infty}f_{m,z_{j}}\lambda^{-m} \\ \lambda^{i}+\lambda\sum\limits_{k=1}^{\infty}g_{k,z_{i}}\lambda^{-k}& \lambda\sum\limits_{m=1}^{\infty}f_{m,z_{i}}\lambda^{-m}
\end{vmatrix}_{+}
\\&=\lambda^{j+1}\sum\limits_{m=1}^{j+1}\frac{f_{m,z_{i}}}{\lambda^{m}}-\lambda^{i+1}\sum\limits_{m=1}^{i+1}\frac{f_{m,z_{j}}}{\lambda^{m}}+g_{1,z_{j}}f_{1,z_{i}}-g_{1,z_{i}}f_{1,z_{j}}
\\&=\begin{vmatrix}\{-\lambda^{j}x+B_{j},\phi_{1}\}&\{-\lambda^{j}x+B_{j},\phi_{2}\} \\ \{-\lambda^{i}x+B_{i},\phi_{1}\}&\{-\lambda^{i}x+B_{i},\phi_{2}\}
\end{vmatrix}_{+}
\\&=\lambda^{i}B_{j,y}-\lambda^{j}B_{i,y}+\{B_{j},B_{i}\}.
\end{align*}
According to the identity
$$f_{m}=\partial^{-1}_{x}(f_{m-1,t}+u_{xy}f_{m-1,x}-u_{xx}f_{m-1,y})$$
and identity (20),
the identity 
$$\left(f_{j+1,z_{i}}-f_{i+1,z_{j}}+g_{1,z_{j}}f_{1,z_{i}}-g_{1,z_{i}}f_{1,z_{j}}\right)_{x}=0$$
can be calculated. Therefore, we can get
$$\frac{\lambda\partial B_{i}}{\partial z_{j}}-\frac{\lambda\partial B_{j}}{\partial z_{i}}+\lambda^{j}B_{i,y}-\lambda^{i}B_{j,y}+\{B_{i},B_{j}\}=0,$$
which is the equation (21e).
On the other hand,
\begin{align*}
[Z_{j},Z_{i}]&=[\lambda\partial_{z_{j}}-\{-\lambda^{j}x+B_{j},\cdot\},\lambda\partial_{z_{i}}-\{-\lambda^{i}x+B_{i},\cdot\}]
\\&=\left(\frac{\lambda\partial B_{i}}{\partial z_{j}}-\frac{\lambda\partial B_{j}}{\partial z_{i}}+\lambda^{j}B_{i,y}-\lambda^{i}B_{j,y}+\{B_{i},B_{j}\}\right)_{x}\partial_{y}
\\&-\left(\frac{\lambda\partial B_{i}}{\partial z_{j}}-\frac{\lambda\partial B_{j}}{\partial z_{i}}+\lambda^{j}B_{i,y}-\lambda^{i}B_{j,y}+\{B_{i},B_{j}\}\right)_{y}\partial_{x}
\\&=0.
\end{align*}
Similarly, we can get equation (21d) and its equivalent commutation condition of $Z_{j},Z_{1}$. 

Next, we will give a related example to demonstrate the equivalence shown in Theorem 3.1.
We only take equation (21c) as an example.
When j=2 and n=2, we can get
\begin{align*} 
&A_{2}=\lambda^{2}\sum\limits_{k=1}^{2}\frac{g_{k}}{\lambda^{k}},
\\&B_{2}=-\lambda^{2}\sum\limits_{m=1}^{2}\frac{f_{m}}{\lambda^{m}}-u_{z_{2}}.
\end{align*}
Therefore, 
\begin{align*} 
&\lambda^{3}\sum\limits_{k=1}^{2}\frac{g_{k,z_{2}}}{\lambda^{k}}+\lambda^{2}\sum\limits_{m=1}^{2}\frac{f_{m,t_{2}}}{\lambda^{m}}+u_{z_{2},t_{2}}+\{\lambda^{2}\sum\limits_{k=1}^{2}\frac{g_{k}}{\lambda^{k}},-\lambda^{2}\sum\limits_{m=1}^{2}\frac{f_{m}}{\lambda^{m}}-u_{z_{2}}\}
\\&=-\lambda^{4}\sum\limits_{k=1}^{2}\frac{g_{k,y}}{\lambda^{k}}+
\lambda^{4}\sum\limits_{m=1}^{2}\frac{f_{m,x}}{\lambda^{m}}+\lambda^{2}u_{z_{2},x}.
\end{align*}
Then we can obtain the following equations as
\begin{align*}
  &g_{2,z_{2}}+u_{yt_{2}}+u_{xy}u_{xz_{2}}+u_{xy}f_{2,x}+u_{xy}g_{2,y}-f_{2,y}u_{xx}-u_{yy}g_{2,x}-u_{yz_{2}}u_{xx}=0,
  \\&f_{2,t_{2}}+u_{z_{2}t_{2}}+f_{2,x}g_{2,y}+u_{xz_{2}}g_{2,y}-f_{2,y}g_{2,x}-u_{yz_{2}}g_{2,x}=0,
\\&f_{2,x}-g_{2,y}+u_{xx}u_{yy}-u_{xy}u_{xy}=0.
\end{align*}
These equations are identical to the flow equations which are obtained from commutation condition of $T_{2}$ and $Z_{2}$ in Example 6. Similarly,  other equations in Theorem 3.1 can be used to calculate the flow equations. 
\end{proof}

\section{\sc \bf The pre-reduced hierarchy  }
In our previous definition of hierarchy, $detJ=1$ serves as a reduction condition. 
When we remove this condition, it allows for a  general hierarchy.
For example, when $detJ=1$, the Manakov-Santini hierarchy can reduce to dKP hierarchy.
The general Dunajski hierarchy  defined by Bogdanov, Dryuma and Manakov
is $((detJ)^{-1}d\phi^{0}\wedge d\phi^{1}\wedge \phi^{2})_{-}=0.$ When $detJ=1$,
the well-known Dunajski hierarchy will be presented. 
Similarly, pre-reduced Doubrov-Ferapontov modified heavenly hierarchy can also be constructed.
\begin{definition}
  The pre-reduced hierarchy associated with the modified heavenly equation is defined as
\begin{equation}\left((det^{-1}J)\Phi_{1}\wedge \Phi_{2}\right)_{-}=0.\end{equation} 
\end{definition}
Similarly we deal with the pre-reduced hierarchy in section 2, we can obtain that

$$
\begin{aligned}
\begin{gathered}
    \begin{pmatrix}
    \phi_{1,t_{n}} ,& \phi_{2,t_{n}}
    \end{pmatrix}
\end{gathered}J^{*}&=
\begin{gathered}
\begin{pmatrix}
 \phi_{1,t_{n}} ,&  \phi_{2,t_{n}}
\end{pmatrix}
\end{gathered}
\begin{gathered}
  \begin{pmatrix}
  \phi_{2,y}& -\phi_{2,x} \\ -\phi_{1,y}&\phi_{1,x}
  \end{pmatrix}
\end{gathered}
=(detJ)
\begin{gathered}
\begin{pmatrix}
X_{n},&Y_{n}
\end{pmatrix}
\end{gathered}
\end{aligned}$$
and 
$$\begin{aligned}
  \begin{gathered}
      \begin{pmatrix}
      \lambda  \phi_{1,z_{j}} ,& \lambda  \phi_{2,z_{j}}
      \end{pmatrix}
  \end{gathered}J^{*}=
  \begin{gathered}
    \begin{pmatrix}
    \lambda  \phi_{1,z_{j}} ,& \lambda  \phi_{2,z_{j}}
    \end{pmatrix}
  \end{gathered}
    \begin{gathered}
      \begin{pmatrix}
      \phi_{2,y}& -\phi_{2,x} \\ -\phi_{1,y}&\phi_{1,x}
      \end{pmatrix}
      \end{gathered}
    =(detJ)(\tilde{X_{j}},\tilde{Y_{j}}),
  \end{aligned}$$
where 
\begin{align*}
&P_{n}=\left((det^{-1}J)\cdot\frac{\partial(\phi_{1},\phi_{2})}{\partial(t_{n},y)}\right)_{+},Q_{n}=\left(-(det^{-1}J)\cdot\frac{\partial(\phi_{1},\phi_{2})}{\partial(t_{n},x)}\right)_{+},
\\&\tilde{P_{j}}=\left((det^{-1}J)\cdot\frac{\partial(\phi_{1},\phi_{2})}{\partial(z_{j},y)}\right)_{+},\tilde{Q_{j}}
=\left(-(det^{-1}J)\cdot\frac{\partial(\phi_{1},\phi_{2})}{\partial(z_{j},x)}\right)_{+}.
\end{align*}
Therefore we can get the following Lax-Sato equations.
\begin{thm}

All coefficient functions $g_{k}$ and $f_{m}$  solve the hierarchy $(22)$ if and only if  they satisfy  the following Lax-Sato form equations 
\begin{align*}
&\partial_{t_{n}}(\overset{\rightharpoonup }{\phi})=
\left(\left((det^{-1}J)\begin{gathered}
\begin{vmatrix}\phi_{1,t_{n}}& \phi_{2,t_{n}} \\ \phi_{1,y}&\phi_{2,y}
\end{vmatrix}
\end{gathered}\right)_{+}\partial_{x}
+
\left((det^{-1}J)\begin{gathered}
\begin{vmatrix}\phi_{2,t_{n}}& \phi_{1,t_{n}} \\ \phi_{2,x}&\phi_{1,x}
\end{vmatrix}
\end{gathered}\right)_{+}\partial_{y}\right)(\overset{\rightharpoonup }{\phi}),\\
&\lambda\partial_{z_{j}}(\overset{\rightharpoonup }{\phi})=
\left(\left((det^{-1}J)\begin{gathered}
\begin{vmatrix}\lambda \phi_{1,z_{j}}& \lambda \phi_{2,z_{j}} \\ \phi_{1,y}&\phi_{2,y}
\end{vmatrix}
\end{gathered}\right)_{+}\partial_{x}
+
\left((det^{-1}J)\begin{gathered}
\begin{vmatrix}\lambda \phi_{2,z_{j}}& \lambda \phi_{1,z_{j}} \\ \phi_{2,x}&\phi_{1,x}
\end{vmatrix}
\end{gathered}\right)_{+}\partial_{y}\right)(\overset{\rightharpoonup }{\phi}).
\end{align*}
\end{thm}
The proof of the theorem is similar to the proof in section 2, we only display the coefficients of $dt_{n} \wedge dz_{j}$ as
$$
\begin{aligned}
(det^{-1}J)*\frac{\lambda \partial(\phi_{1},\phi_{2})}{\partial(t_{n},z_{j})}
&=(det^{-1}J)\begin{vmatrix}\phi_{1,t_{n}}& \phi_{2,t_{n}} \\ \lambda\phi_{1,z_{j}}&\lambda\phi_{2,z_{j}}\end{vmatrix}
\\&=(det^{-1}J)\begin{vmatrix}[P_{n}]_{+}\phi_{1,x}+[Q_{n}]_{+}\phi_{1,y}& [P_{n}]_{+}\phi_{2,x}+[Q_{n}]_{+}\phi_{2,y} 
  \\ [\tilde{P_{n}}]_{+}\phi_{1,x}+[\tilde{Q_{n}}]_{+}\phi_{1,y}&[\tilde{P_{n}}]_{+}\phi_{2,x}+[\tilde{Q_{n}}]_{+}\phi_{2,y}\end{vmatrix}
\\&=(det^{-1}J)\begin{vmatrix}[P_{n}]_{+}& [Q_{n}]_{+} \\ [\tilde{P_{n}}]_{+}&[\tilde{Q_{n}}]_{+}\end{vmatrix}.
\begin{vmatrix}
  \phi_{1,x}&\phi_{2,x}\\
  \phi_{1,y}&\phi_{2,y}
  \end{vmatrix}
=\begin{vmatrix}[P_{n}]_{+}& [Q_{n}]_{+} \\ [\tilde{P_{n}}]_{+}&[\tilde{Q_{n}}]_{+}\end{vmatrix},
\end{aligned}
$$
so 
$$\left((det^{-1}J)\frac{\lambda \partial(\phi_{1},\phi_{2})}{\partial(t_{n},z_{j})}\right)_{-}=0.$$

\bigskip
\section{\sc \bf The nonlinear Riemann-Hilbert problem of modified heavenly equation }
Due to the fact that the space of the eigenfunctions of the vector fields is a ring,  the inverse problem is essentially nonlinear and this nonlinear problem  can be formulated as a nonlinear Riemann-Hilbert problem
on a suitable contour of the complex plane of the spectral parameter.
In the field related to dispersionless integrable equations, there exists novel IST named Manakov-Santini method.
This method can be applied to study the long-time behavior of dispersionless integrable equations.
 Linking dispersionless integrable equations to the nonlinear Riemann-Hilbert problem is the first step in applying the IST.
We first consider
the vector nonlinear Riemann-Hilbert problem on the real line as
\begin{equation}
\begin{gathered}
  \begin{pmatrix}\phi^{+}_{1}(\lambda) \\ \phi^{+}_{2}(\lambda)
  \end{pmatrix}
  \end{gathered}
  =\begin{gathered}
    \begin{pmatrix}\phi^{-}_{1}(\lambda)+R_{1}(\phi^{-}_{1}(\lambda),\phi^{-}_{2}(\lambda),\lambda) \\ \phi^{-}_{2}(\lambda)+R_{2}(\phi^{-}_{1}(\lambda),\phi^{-}_{2}(\lambda),\lambda)
    \end{pmatrix}
    \end{gathered}
    =\begin{gathered}
      \begin{pmatrix}\mathcal{R}_{1} (\phi^{-}_{1}(\lambda),\phi^{-}_{2}(\lambda),\lambda)\\ \mathcal{R}_{2} (\phi^{-}_{1}(\lambda),\phi^{-}_{2}(\lambda),\lambda)
      \end{pmatrix}
      \end{gathered}
    ,\lambda\in \mathbb{R} ,
  \end{equation}
   where the solutions $\phi^{\pm}_{1},\phi^{\pm}_{2}$ are analytic in the upper and lower halves of the complex
   $\lambda$ plane, and they are normalized as 
   \begin{equation}
      \begin{gathered}
        \begin{pmatrix}\phi^{\pm}_{1}(\lambda) \\ \phi^{\pm}_{2}(\lambda)
        \end{pmatrix}
        \end{gathered}=
          \begin{gathered}
            \begin{pmatrix}-y+O   (\lambda^{-1}) \\ x+\lambda t+O   (\lambda^{-1})
            \end{pmatrix}
            \end{gathered},
            |\lambda|\gg 1,
   \end{equation}
 where $\overset{\rightharpoonup }{R}(\overset{\rightharpoonup }{\zeta},\lambda)=(R_{1}(\zeta_{1},\zeta_{2},\lambda),R_{2}(\zeta_{1},\zeta_{2},\lambda))^{T}$are
spectral data that defined for $\overset{\rightharpoonup }{\zeta}\in \mathbb{C} ^{2}$ and $\overset{\rightharpoonup }{\mathcal{R}}=\overset{\rightharpoonup }{\zeta}+(\overset{\rightharpoonup }{R}(\overset{\rightharpoonup }{\zeta},\lambda)).$
\begin{proposition}
Assuming that the above Riemann-Hilbert problem and its linearized form $\overset{\rightharpoonup }{\sigma }^{+}=J\overset{\rightharpoonup }{\sigma }^{-}$are uniquely solvable, where
$J$ is the Jacobian matrix of the transformation $(23)$ and $J_{ij}=\partial_{i} \mathcal{R}/\partial\zeta_{j} ,i,j=1,2$, the solutions $\phi^{\pm }(\lambda)$
of the Riemann-Hilbert problem $(23)$ are common eigenfunctions on the vector fields $(2)$ and 
\begin{align}
  &u_{x}=\lim_{\lambda\rightarrow \infty}(\lambda (\phi_{1}^{\pm}+y)),
  \\&u_{y}=\lim_{\lambda\rightarrow \infty}(\lambda (\phi_{2}^{\pm}-x-\lambda t)).
\end{align}
\end{proposition}

\begin{proposition}
 If the spectral data $\overset{\rightharpoonup }{\mathcal{R}}(\overset{\rightharpoonup }{\zeta},\lambda)$ satisfies the reality constraint as
\begin{equation}
  \overset{\rightharpoonup }{\mathcal{R}}(\overline{\overset{\rightharpoonup }{\mathcal{R}}(\overline{\overset{\rightharpoonup }{\zeta}, }\lambda)},\lambda )=\overset{\rightharpoonup }{\zeta},
\end{equation}
the solution u is real and satisfies the heavenly constraint as
\begin{equation} 
  \{ \mathcal{R}_{1},\mathcal{R}_{2}\} _{(\zeta_{1},\zeta_{2})}=1,
\end{equation}where $\{\cdot , \cdot \} $ is the Poisson bracket and $\overline{\overset{\rightharpoonup }{\zeta}}$ is conjugate of $\overset{\rightharpoonup }{\zeta}$.
\end{proposition}
\begin{proof}
  We apply the operators $T,Z$ in (2) to the Riemann-Hilbert problem (23). According to chain rule, we can get
  $$T_{1}(\phi^{+}_{1})=\frac{\partial\mathcal{R}_{1}}{\partial \phi^{-}_{1}}\cdot T_{1}(\phi^{-}_{1})+\frac{\partial \mathcal{R}_{1}}{\partial \phi^{-}_{2}}\cdot T_{1}(\phi^{-}_{2})$$ and other similar examples.
  So $T_{1}(\phi^{\pm})$ and  $Z_{1}(\phi^{\pm})$ satisfy the linearized form
  $$
  \begin{gathered}
    \begin{pmatrix}T_{1}(\phi^{+}_{1}(\lambda)) &Z_{1}(\phi^{+}_{1}(\lambda))\\ T_{1}(\phi^{+}_{2}(\lambda)) &Z_{1}(\phi^{+}_{2}(\lambda))
    \end{pmatrix}
    \end{gathered}
    =\begin{gathered}
      \begin{pmatrix}\frac{\partial \mathcal{R}_{1}}{\partial \phi_{1}^{-}} &\frac{\partial \mathcal{R}_{1}}{\partial \phi_{2}^{-}} \\ \frac{\partial \mathcal{R}_{2}}{\partial \phi_{1}^{-}}&\frac{\partial \mathcal{R}_{2}}{\partial \phi_{2}^{-}}
      \end{pmatrix}
      \end{gathered}
      \begin{gathered}
        \begin{pmatrix}T_{1}(\phi^{-}_{1}(\lambda)) &Z_{1}(\phi^{-}_{1}(\lambda))\\ T_{1}(\phi^{-}_{2}(\lambda)) &Z_{1}(\phi^{-}_{2}(\lambda))
        \end{pmatrix}
        \end{gathered}.
      $$
    Apart from that, we can also get the identities (25) and (26). Therefore we can find that $T_{1}(\phi^{\pm})\to 0, Z_{1}(\phi^{\pm})\to 0$, as $\lambda \to \infty$, and by uniqueness we
    infer that $\vec{\phi}^{\pm}$ are eigenfunctions of the vector fields $T_{1},Z_{1}$. 
   Then we assume that $\overline{\overset{\rightharpoonup }{\phi^{+} }}=\overset{\rightharpoonup }{\phi^{-} }$ holds, if identity (27) is imposed and due to (23),
   $$ \overset{\rightharpoonup }{\mathcal{R}}(\overline{\overset{\rightharpoonup }{\mathcal{R}}(\overline{\overset{\rightharpoonup }{\phi^{+} }, }\lambda)},\lambda )=\overset{\rightharpoonup }{\mathcal{R}}(\overset{\rightharpoonup }{\phi^{+} }+\overline{\overset{\rightharpoonup }{\mathcal{R}}(\overset{\rightharpoonup }{\phi^{-} },\lambda)})=\overset{\rightharpoonup }{\phi^{+} },$$
   by uniqueness, this assumption is correct. In other words, the solution is reality.
   Finally, if the heavenly constraint (28) is imposed, the identities can be calculated as
   $$\{\phi_{1}^{+},\phi_{2}^{+} \}=\{ \mathcal{R}_{1},\mathcal{R}_{2}\}=\{ \phi_{1}^{-},\phi_{2}^{-} \},  $$
  and due to identity (24), we can get $$\{ \phi_{1}^{\pm},\phi_{2}^{\pm} \}=1, |\lambda| \to \infty.$$
   According to the analyticity of the eigenfunctions, we can get 
   $$\{\phi_{1}^{+},\phi_{2}^{+} \}=\{ \phi_{1}^{-},\phi_{2}^{-} \}=1.$$
\end{proof}

\bigskip
\bigskip
\textbf{Acknowledgements:} This work is supported by the National Natural Science Foundation of China under Grant Nos. 12271136, 12171133 and 12171132.

\end{document}